\documentclass[12pt]{article}
\usepackage{amsfonts}
\usepackage[utf8]{inputenc} 
\usepackage{amsmath}
\usepackage{amssymb}
\usepackage{tikz}

\newtheorem{theorem}{Theorem}
\newtheorem{lemma}[theorem]{Lemma}
\newenvironment{proof}[0]{\textit{Proof.~}}{\hfill \( \Box \) \vspace{2mm} }

\title{On the Integrality Ratio of the Subtour LP for Euclidean TSP}
\author{Stefan Hougardy\\Research Institute for Discrete Mathematics, University of Bonn,\\
            Lenn\'estr.~2, 53113 Bonn, Germany}

\begin{document}

\maketitle 

\begin{abstract}
A long standing conjecture says that the integrality ratio of the subtour LP for metric TSP is $4/3$. 
A well known family of graphic TSP instances achieves this lower bound asymptotically. 
For Euclidean TSP the best known lower bound on the integrality ratio was $8/7$. 
We improve this value by presenting 
a family of Euclidean TSP instances for which the integrality ratio of the subtour 
LP converges to 4/3. 
\end{abstract}

{\small\textbf{keywords:} traveling salesman problem; subtour LP; Held-Karp bound; integrality ratio; 
Euclidean TSP}

\section{Introduction}

Given $n$ cities with their pairwise distances the \emph{traveling salesman problem} (TSP) 
asks for a shortest tour that visits each city exactly once.
This problem is known to be NP-hard~\cite{Kar1972} and cannot be approximated
within any constant factor unless P=NP~\cite{SG1976}.
For \emph{metric TSP} instances, i.e., TSP instances where
the distances satisfy the triangle inequality, 
a $3/2$ approximation algorithm was presented by Christofides in 1976~\cite{Chr1976}.
In spite of many efforts no improvement on this approximation ratio for the general metric
TSP has been achieved so far. 
One approach to obtain a better approximation algorithm for the metric TSP is 
based on the \emph{subtour LP}. This LP is a relaxation of an integer program for the TSP
that was first used by Dantzig, Fulkerson, and Johnson in 1954~\cite{DFJ1954}.
If the cities are numbered from $1$ to $n$ and $c_{ij}$ denotes the distance between city $i$ and city $j$ then
the subtour LP can be formulated as follows:\\

\noindent\framebox{
\parbox{0.96\linewidth}{\vspace*{2mm}\hspace*{10mm}
$\displaystyle \text{minimize: } \sum_{1\le i<j\le n} c_{ij}\cdot x_{ij}$ 
\begin{eqnarray*}
\hspace*{-5mm}\text{subject to:}\\[3mm]
             \sum_{j\not = i} x_{ij} & = & 1 ~~\text{ for all } i \in \{1,\ldots,n\} \\
             \sum_{j\not = i} x_{ji} & = & 1 ~~\text{ for all } i \in \{1,\ldots,n\} \\
\sum_{i,j\in S} x_{ij}    & \le & |S| - 1 ~~\text{ for all }  \emptyset\not = S\subsetneq \{1,\ldots,n\}\\
0 & \le & x_{ij} ~\le~ 1\\[-6mm]
\end{eqnarray*}}}\\[6mm]

This LP has an exponential number of constraints but can be solved in 
polynomial time via the ellipsoid method as the separation problem can be solved efficiently~\cite{GLS1981}. 
The \emph{integrality ratio} of the subtour LP for the metric TSP is the supremum 
of the length of an optimum TSP tour over the optimum solution of the subtour LP.
Wolsey~\cite{Wol1980} has shown that the integrality ratio of the subtour LP for metric TSP is at most $3/2$.
A well known conjecture states that the integrality ratio of the subtour LP for metric TSP
is $4/3$. This conjecture seems to be mentioned for the first in 1990~\cite[page 35]{Wil1990} but according 
to~\cite{Goe2012} it was already well known several years before. 
A proof of this conjecture yields a polynomial time algorithm
that approximates the value of an optimum TSP tour within a factor of $4/3$. 
It is known that the integrality ratio of the subtour LP is at least $4/3$ as there exists
a family of metric TSP instances whose integrality ratio converges to $4/3$~(see for example~\cite{Wil1990}).

For the metric TSP the lower and upper bound of $4/3$ and $3/2$ on the integrality ratio 
of the subtour LP have not been improved for more than 25 years. Therefore, people became 
interested to study the integrality ratio of the subtour LP for special cases of the metric TSP.

The \emph{graphic TSP} is a special case of the metric TSP where the distances between the $n$ cities
are the lengths of shortest paths in an undirected connected graph on these $n$ cities. 
For graphic TSP the integrality ratio
of the subtour LP is at least $4/3$~\cite{Wil1990} and at most 1.4~\cite{SV2012}.

In the $1$-$2$-TSP the distances between two cities are either 1 or 2. Here the largest known lower bound on the 
integrality ratio of the subtour LP is $10/9$~\cite{Wil1990} while the smallest known upper bound is 
$5/4$~\cite{QSWZ2014}. 

In the \emph{Euclidean TSP} the cities are points in the plane and their distance is the Euclidean
distance between the two points. In this case the best known lower bound on the integrality ratio
seems to be $8/7$ (mentioned in~\cite{Wol1980}). The best known upper bound is as in the general metric case $3/2$.

In this paper we will improve the lower bound for the integrality ratio of the Euclidean TSP by presenting
a family of Euclidean TSP instances for which the integrality ratio of the subtour LP converges to $4/3$. 
We will prove this result in Section~\ref{sec:construction}. Using a more careful analysis we prove an
explicit formula for the integrality ratio in Section~\ref{sec:ExactFormula}.

\section{Euclidean instances with integrality ratio 4/3}\label{sec:construction}

We will now describe our construction of a family of Euclidean TSP instances 
for which the integrality ratio of the subtour LP converges to $4/3$. 
Each instance of the family contains equidistant points on three parallel lines. 
More precisely, it contains the $3 n$ points with coordinates 
$(i, j\cdot d)$ for $i=1, \ldots, n$ and $j=1,2,3$ where $d$ is the distance between the parallel lines.
We will denote these instances by $G(n,d)$.  Figure~\ref{fig:G(18,3)} shows the instance $G(18,3)$.

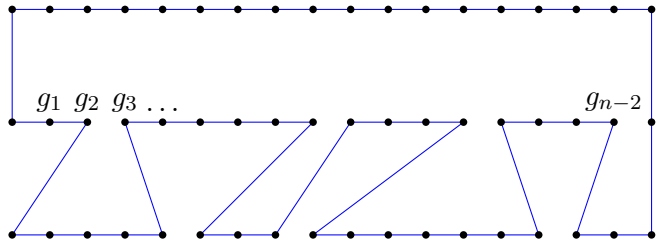
\begin{figure}[ht]
\centering
\begin{tikzpicture}[scale=0.5]

\foreach \x in {0,1,2,3,4,5,6,7,8,9,10,11,12,13,14,15,16}
   \draw[blue] (\x, 6) -- (\x+1, 6);

\draw[blue] ( 0, 6) -- ( 0, 3);
\draw[blue] (17, 6) -- (17, 3);

\foreach \x in {0,1,3,4,5,6,7,9,10,11,13,14,15}
   \draw[blue] (\x, 3) -- (\x+1, 3);

\foreach \x in {0,1,2,3,5,6,8,9,10,11,12,13,15,16}
   \draw[blue] (\x, 0) -- (\x+1, 0);
   
\draw[blue] ( 2, 3) -- ( 0, 0);
\draw[blue] ( 3, 3) -- ( 4, 0);
\draw[blue] ( 8, 3) -- ( 5, 0);
\draw[blue] ( 9, 3) -- ( 7, 0);
\draw[blue] (12, 3) -- ( 8, 0);
\draw[blue] (13, 3) -- (14, 0);
\draw[blue] (16, 3) -- (15, 0);
\draw[blue] (17, 3) -- (17, 0);

\foreach \x in {0,1,2,3,4,5,6,7,8,9,10,11,12,13,14,15,16,17}
	\foreach \y in {0,3,6}
		\fill (\x, \y) circle(1mm);

\foreach \x in {1,2,3}
  \draw (\x,3) node[anchor = south] {\small $g_{\x}$};
\draw (4,3) node[anchor = south] {\small $\ldots$};
\draw (16,3) node[anchor = south] {\small $g_{n-2}$};

\end{tikzpicture}
\caption{A TSP tour for the instance $G(18,3)$.}
\label{fig:G(18,3)}
\end{figure}

The instances $G(n,d)$ belong to the class of so called \textit{convex-hull-and-line TSP}.
These are Euclidean TSP instances where all points that do not lie on the boundary of the convex hull 
of the point set lie
on a single line segment inside the convex hull.     
Deineko, van Dal, and Rote~\cite{DDR1994} have shown that
an optimum TSP tour for a convex-hull-and-line TSP can be found in polynomial time. 

Following the notation in~\cite{DDR1994} 
we denote the point $(i+1, 2 \cdot d)$ in the instance $G(n,d)$ by $g_i$ for $i=1, \ldots, n-2$
(see Figure~\ref{fig:G(18,3)}).
Thus, the set ${\cal G} := \{g_1, \ldots, g_{n-2}\}$ contains all points of the instance $G(n,d)$ that do not lie
on the boundary of the convex hull of $G(n,d)$. Let $\cal B$ denote the set of all other points in 
$G(n,d)$.  An optimum TSP tour for $\cal B$ is obtained by visiting the points in $\cal B$
in their cyclic order~\cite{DDR1994}. Moreover, this tour is unique. Therefore, we can call two points in $\cal B$
\textit{adjacent} if they are adjacent in the optimum tour for $\cal B$. Let ${\cal B}_l$ denote all points
in $\cal B$ that lie on the two lower lines.  
The following structural result is a special case of Lemma~3 in~\cite{DDR1994}:

\begin{lemma}[Lemma~3 in \cite{DDR1994}]
An optimum TSP tour for the instance $G(n,d)$ can be obtained by splitting the set of points
$\cal G$ into $k+1$ segments 
$$\{g_1, g_2, \ldots, g_{i_1}\}, ~~~
\{g_{i_1+1}, \ldots, g_{i_2}\}, ~~~ \ldots, ~~~
\{g_{i_k+1}, \ldots, g_{n-2}\}$$
for $0\le k \le n-2$, $0=i_0 < i_1 < i_2 < \cdots < i_k < n-2$,
and inserting each segment between two adjacent points in ${\cal B}_l$. 
\end{lemma}

\begin{proof} 
This is exactly the statement of Lemma~3 in \cite{DDR1994} except that 
$\cal B$ is replaced by ${\cal B}_l$.
Thus, we only have to observe that because of symmetry of the instance $G(n,d)$ 
we may assume that all segments are inserted into adjacent points in ${\cal B}_l$.
\end{proof}

Let the \textit{cost} of inserting a segment $\{g_i, g_{i+1}, \ldots, g_j\}$ into $\cal B$ 
be the difference in the length of the optimum tour after and before inserting this segment.
From Lemma~4 in~\cite{DDR1994} it follows that a segment $\{g_i, g_{i+1}, \ldots, g_j\}$
that contains neither $g_1$ nor $g_{n-2}$ must be inserted between two adjacent points from the lower
line of points in $G(n,d)$. As the points have unit distance on the lines, the cost of inserting such a segment
only depends on the number of points contained in the segment.

\begin{lemma}\label{lemma:insertioncost}
For $d\ge 4$ the cost of inserting a segment of $i$ points from $\cal G$ into $\cal B$ is at least
$$ \begin{cases} 
    i-2 + \max\{2d, i-2\} & \text{if $g_1$ and $g_{n-2}$ do not belong to the segment}\\
    i-d + \max\{d, i\}    & \text{if $g_1$ or $g_{n-2}$ belong to the segment}
  \end{cases}$$
\end{lemma}

\begin{proof}
As already mentioned above it follows from Lemma~4 in~\cite{DDR1994} that a segment 
that contains neither $g_1$ nor $g_{n-2}$ must be inserted between two adjacent points from the lower
line of points in $G(n,d)$. 
The total cost of inserting such a segment is $i-1$ for the horizontal connection of the
$i$ points in the segment plus the two connections from the end points of the segment to two adjacent points in
the lower line of points in $G(n,d)$ minus 1 (for the edge that is removed between the two adjacent points in 
$\cal B$).
Thus we get a lower bound of $i-2 + \max\{2d, i-2\}$.

If $g_1$ or $g_{n-2}$ is contained in the segment then there are two possibilities to 
insert the segment which are shown in Figure~\ref{fig:insert_g1}.
For case b) we get again the lower bound $i-2 + \max\{2d, i-2\}$ while in case a) we have a lower bound
of $i-d+\max\{d, i\}$. 
Since $i-2 + \max\{2d, i-2\} \ge i-d+\max\{d, i\}$ for $d\ge 4$ the result follows. 
\end{proof}

\begin{figure}[ht]
\centering
\begin{tabular}{c@{\hspace*{1.2cm}}c}
\begin{tikzpicture}[scale=0.5]
\draw[blue] ( 0, 6) -- ( 0, 3);

\foreach \x in {0,1,2,3,4,5}
   \draw[blue] (\x, 3) -- (\x+1, 3);

\foreach \x in {0,1,2,3,4,5,6}
   \draw[blue] (\x, 0) -- (\x+1, 0);

\draw[blue] ( 6, 3) -- ( 0, 0);

\foreach \x in {0,1,2,3,4,5,6}
	\foreach \y in {0,3}
		\fill (\x, \y) circle(1mm);

\foreach \x in {1,2,3}
  \draw (\x,3) node[anchor = south] {\small $g_{\x}$};
\draw (4,3) node[anchor = south] {\small $\ldots$};
\end{tikzpicture}&
\begin{tikzpicture}[scale=0.5]
\draw[blue] ( 0, 6) -- ( 0, 3);
\draw[blue] ( 0, 3) -- ( 0, 0);

\foreach \x in {1,2,3,4,5}
   \draw[blue] (\x, 3) -- (\x+1, 3);

\foreach \x in {0,1,2,4,5,6}
   \draw[blue] (\x, 0) -- (\x+1, 0);

\draw[blue] ( 1, 3) -- ( 3, 0);
\draw[blue] ( 6, 3) -- ( 4, 0);

\foreach \x in {0,1,2,3,4,5,6}
	\foreach \y in {0,3}
		\fill (\x, \y) circle(1mm);

\foreach \x in {1,2,3}
  \draw (\x,3) node[anchor = south] {\small $g_{\x}$};
\draw (4,3) node[anchor = south] {\small $\ldots$};
\end{tikzpicture}\\[2mm]
a) & b)
\end{tabular}
\caption{The two possibilities to insert into $\cal B$ a segment containing the point $g_1$.}
\label{fig:insert_g1}
\end{figure}
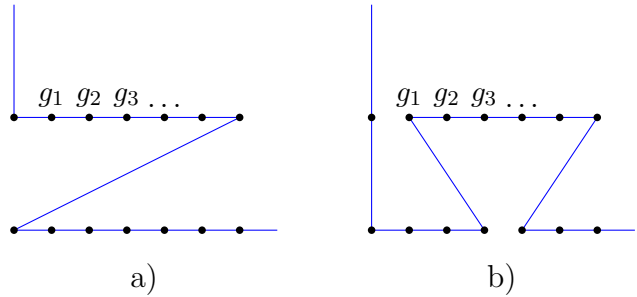

We now can compute a lower bound on the length of an optimum TSP tour for $G(n,d)$ as follows.

\begin{lemma}\label{lemma:approxlength}
Let $d\ge 4$. Then an optimum TSP tour for $G(n,d)$ has length at least 
$4n + 2d-2-2n/(d+1)$.
\end{lemma}

\begin{proof}
For $k$ with $1\le k \le n-2$ let $z_1, z_2, \ldots, z_k$ be the number of points contained in the segments of $\cal G$ 
that are inserted into $\cal B$ for an optimum TSP tour of $G(n,d)$. 
We have $\sum_{i=1}^k z_i = n-2$.
The boundary of the convex hull for the point set $\cal B$ has length $2n-2+4d$. 
 
By Lemma~\ref {lemma:insertioncost} 
the total length of the tour is at least 
\begin{eqnarray*}
 2n-2+4d+\sum_{i=1}^{k} (z_i-2+2d) -4d + 4 & = & 3n + 2k (d-1) 
\end{eqnarray*} 
where the term $-4d + 4\le 0$ is for adjusting for the at most two segments containing $g_1$ and $g_{n-2}$.

On the other hand by Lemma~\ref {lemma:insertioncost} 
the total length of the tour is at least 
\begin{eqnarray*}
 2n-2+4d+\sum_{i=1}^{k} (2z_i-4) - 2d + 8 & = & 4n + 2 - 4k + 2d 
\end{eqnarray*} 
Here we need the term $-2d+8\le 0$ to adjust for
the at most two segments containing $g_1$ and $g_{n-2}$.

This shows that an optimum TSP tour for $G(n,d)$ has length at least
\begin{equation}
\min_{1\le k\le n-2} \max\{3n + 2k (d-1), 4 (n - k) + 2d + 2\}.
\label{eqn:twolowerbounds}
\end{equation}
The two functions $3n + 2k (d-1)$ and $4 (n - k) + 2d + 2$ are both linear in $k$ 
and the slopes have opposite sign. Thus, the expression (\ref{eqn:twolowerbounds})
is at least as large as the value at the intersection of these two linear functions.
The two linear functions intersect at $k=1+n/(2d+2)$. The value at the intersection point
is $4(n-1-n/(2d+2))+2d+2 = 4n +2d - 2 - 2n/(d+1)$.
This finishes the proof.
\end{proof}

Now we can state and prove our main result.

\begin{theorem}\label{thm:main}
Let $d$ be a function with $d(n) = \omega(1)$ and $d(n) = o(n)$. 
Then the integrality ratio of the subtour LP for the instances $G(n, d(n))$
converges to $4/3$ for $n\to \infty$. 
\end{theorem}
\begin{proof}
As $d(n) = \omega(1)$ we may assume that $d(n) \ge 4$.
By Lemma~\ref{lemma:approxlength} the length of any TSP tour for $G(n, d(n))$ is at least
$4n + 2d-2-2n/(d+1)$.

Figure~\ref{fig:newsubtourbound} shows a feasible solution for the subtour LP for $G(n, d(n))$ with
cost $3n-4 + 3 d(n) + \sqrt{d(n)^2 + 1} \le 3n + 4 d(n)$.
Thus, the integrality ratio for $G(n, d(n))$ is at least
\begin{equation}\label{eq:lowerbound}
 \frac {4n + 2d(n)-2-2n/(d(n)+1)}{3n + 4 d(n)}
\end{equation}
 which tends to $4/3$ for $n\to \infty$
as $d(n) = \omega(1)$ and $d(n) = o(n)$.
\end{proof}

\begin{figure}[ht]
\centering
\begin{tikzpicture}[scale=0.5]
\foreach \x in {0,1,2,3,4,5,6,7,8,9,10,11,12,13,14,15,16}
   \draw[blue] (\x, 6) -- (\x+1, 6);

\foreach \x in {1,2,3,4,5,6,7,8,9,10,11,12,13,14,15}
   \draw[blue] (\x, 3) -- (\x+1, 3);

\foreach \x in {0,1,2,3,4,5,6,7,8,9,10,11,12,13,14,15,16}
   \draw[blue] (\x, 0) -- (\x+1, 0);

\draw[blue] ( 0, 6) -- ( 0, 3);
\draw[blue] (17, 6) -- (17, 3);

\draw[blue, dashed] ( 0, 3) -- ( 0, 0);
\draw[blue, dashed] ( 1, 3) -- ( 0, 0);
\draw[blue, dashed] ( 0, 3) -- ( 1, 3);
\draw[blue, dashed] (17, 3) -- (17, 0);
\draw[blue, dashed] (16, 3) -- (17, 0);
\draw[blue, dashed] (16, 3) -- (17, 3);

\foreach \x in {0,1,2,3,4,5,6,7,8,9,10,11,12,13,14,15,16,17}
	\foreach \y in {0,3,6}
		\fill (\x, \y) circle(1mm);
\end{tikzpicture}
\caption{A feasible solution to the subtour LP for the instance $G(18,3)$. 
The dashed lines correspond to variables with value~$1/2$ while all other lines 
correspond to variables with value~1.}
\label{fig:newsubtourbound}
\end{figure}
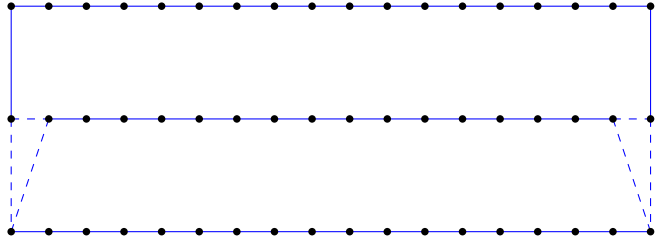

The proofs of Lemma~3 and Lemma~4 in~\cite{DDR1994} make only use of the fact
that optimum TSP tours for certain subsets of the point set do not intersect itself. 
This property holds in case of the instance $G(n,d)$ for all $L^p$-norms with $p\in\mathbb{N}$.
Moreover, the proofs of Lemma~\ref{lemma:insertioncost}
and Lemma~\ref{lemma:approxlength} also hold for arbitrary $L^p$-norms.
Therefore, Theorem~\ref{thm:main} holds for all $L^p$-norms with $p\in\mathbb{N}$,
as $3n+4d(n)$ is an upper bound for an optimum solution for the subtour LP for
any $L^p$-norm .

\section{The exact integrality ratio for $G(n, \sqrt{n-1})$}\label{sec:ExactFormula}

The lower bound~(\ref{eq:lowerbound}) proven in Section~\ref{sec:construction} attains its maximum
for $d(n) = \Theta(\sqrt{n})$. In this section we prove an explicit formula
for the integrality ratio of the instances $G(n, \sqrt{n-1})$. This leads to an improved convergence of
the proven lower bound. The structural results for an optimum TSP tour in the instances 
$G(n,d)$ proven in this section may be of independent interest. 

\begin{lemma}\label{lemma:cheapestinsertion}
For the instance $G(n,d)$ the cheapest cost of inserting into $\cal B$ a segment of $k$ points of $\cal G$ 
that contains neither $g_1$ nor $g_{n-2}$ is
$$ \begin{cases} k-2+\sqrt{(k-2)^2+4\cdot d^2} & \text{if $k$ is even}\\[2mm]
                         k-2+ \frac12 \cdot\left(\sqrt{(k-1)^2+4\cdot d^2} + \sqrt{(k-3)^2+4\cdot d^2}\right)
                         & \text{if $k$ is odd}
  \end{cases}$$
\end{lemma}

\begin{proof}
By shifting the instance $G(n,d)$ we may assume that the $k$ points in the segment have the
coordinates $(i, d)$ for $i=1, \ldots, k$.
The cost of inserting this segment between the two points $(a,0)$ and $(a+1, 0)$
with $a\in\{1,\ldots,k\}$ is $(k-1) + \sqrt{(a-1)^2+d^2} + \sqrt{(k-a-1)^2+d^2} - 1$.
The minimum is attained for $a=\lfloor \frac k 2 \rfloor$.
\end{proof}

If a segment of $\cal G$ contains $g_1$ or $g_{n-2}$ then there exist two different possibilities for
inserting it into $\cal B$, see Figure~\ref{fig:insert_g1}. The following lemma states, that
if $d$ is sufficiently large then possibility a) will be cheaper. 
For the proof of this lemma and in some other proofs we will make use of the inequality
\begin{equation}\label{eq:sqrt}
a+\sqrt{b^2+c} ~ \ge ~ \sqrt{(a+b)^2 +c} ~~~~~ \text{ for all } a,b,c \in \mathbb{R_+}
\end{equation} 

\begin{lemma}\label{lemma:cheapest_insertion_g1}
Let $d\ge 4$.
For the instance $G(n,d)$ the cheapest cost of inserting into $\cal B$ a segment of $k$ points of $\cal G$ 
that contains either $g_1$ or $g_{n-2}$ is
$$ k-d+\sqrt{k^2+d^2}.$$
\end{lemma}

\begin{proof}
There exist two different possibilities to insert into $\cal B$ a segment of $k$ points of $\cal G$ 
that contains the point $g_1$. These two possibilities are shown in Figure~\ref{fig:insert_g1}a) and b).
In case a) the cost of inserting the segment is $k-d+\sqrt{k^2+d^2}.$
In case b) the cost of insertion is by Lemma~\ref{lemma:cheapestinsertion} at least
$k-2+\sqrt{(k-2)^2+4\cdot d^2}$.
We claim that for $d\ge 4$ we have  $k-d+\sqrt{k^2+d^2} \le k-2+\sqrt{(k-2)^2+4\cdot d^2}$.
This is equivalent to the statement that for $d\ge 4$ we have 
$$ h(d) ~:=~ d-2+\sqrt{(k-2)^2+4\cdot d^2}-\sqrt{k^2+d^2} ~\ge~ 0. $$ 
Now, $h(4) =  2+\sqrt{(k-2)^2+64}-\sqrt{k^2+16} \ge \sqrt{k^2+64}-\sqrt{k^2+16} \ge 0$ by (\ref{eq:sqrt}).
Moreover, we have $h'(d) = \frac{4 d}{\sqrt{4 d^2+(k-2)^2}}-\frac{d}{\sqrt{d^2+k^2}}+1 > 0$
which proves the claim. ~~ ~ ~\hspace*{1cm}
\end{proof}

From Lemma~\ref{lemma:cheapestinsertion} and Lemma~\ref{lemma:cheapest_insertion_g1} it
follows that for $d\ge 4$ there always exists an optimum tour of $G(n,d)$ 
which has a \textit{z-structure} as shown in Figure~\ref{fig:z-structure}. This means
that the tour consists of a sequence of alternately oriented z-shaped paths that cover all
points in the lower two lines. These z-shaped paths are connected by single edges of length~1
and the tour is closed by adding the top most horizontal line and two vertical connections to 
the middle line. Each tour with such a z-structure can be specified by a \textit{z-vector}
which contains as entries the number of points covered by each z on the middle line.

\begin{figure*}[ht]
\centering
\begin{tikzpicture}[scale=0.496]
\tikzstyle{help lines}=[gray!90,very thin]

\foreach \x in {0,1,2,3,4,5,6,7,8,9,10,11,12,13,14,15,16,17,18,19,20,21,22,23,24,25,26}
   \draw[blue] (\x, 6) -- (\x+1, 6);

\draw[blue] ( 0, 6) -- ( 0, 3);
\draw[blue] (27, 6) -- (27, 3);

\foreach \x in {0,1,2, 4,5,6,7,8, 10,11,12,13,14,15,16 ,18,19,20,21,22,23, 25,26}
   \draw[blue] (\x, 3) -- (\x+1, 3);

\foreach \x in {0,1,2,3,4,5,  7,8,9,10,11,12, 14,15,16,17,18,19 ,21,22,23,24,25,26}
   \draw[blue] (\x, 0) -- (\x+1, 0);
   
\foreach \x in {0,1,2, 4,5, 7,8, 10,11,12, 14,15,16 ,18,19, 21,22,23, 25,26}
   \draw[blue, thick] (\x, 3) -- (\x+1, 3);

\foreach \x in {0,1,2 ,4,5,  7,8, 10,11,12, 14,15,16, 18,19 ,21,22,23 ,25,26}
   \draw[blue, thick] (\x, 0) -- (\x+1, 0);

\draw[blue, thick] ( 3, 3) -- ( 0, 0);
\draw[blue, thick] ( 4, 3) -- ( 6, 0);
\draw[blue, thick] ( 9, 3) -- ( 7, 0);
\draw[blue, thick] (10, 3) -- (13, 0);
\draw[blue, thick] (17, 3) -- (14, 0);
\draw[blue, thick] (18, 3) -- (20, 0);
\draw[blue, thick] (24, 3) -- (21, 0);
\draw[blue, thick] (25, 3) -- (27, 0);

\draw[style=help lines] ( 0,-0.5) -- ( 0,3.5);
\draw[style=help lines] ( 3,-0.5) -- ( 3,3.5);
\draw[style=help lines] ( 4,-0.5) -- ( 4,3.5);
\draw[style=help lines] ( 6,-0.5) -- ( 6,3.5);
\draw[style=help lines] ( 7,-0.5) -- ( 7,3.5);
\draw[style=help lines] ( 9,-0.5) -- ( 9,3.5);
\draw[style=help lines] (10,-0.5) -- (10,3.5);
\draw[style=help lines] (13,-0.5) -- (13,3.5);
\draw[style=help lines] (14,-0.5) -- (14,3.5);
\draw[style=help lines] (17,-0.5) -- (17,3.5);
\draw[style=help lines] (18,-0.5) -- (18,3.5);
\draw[style=help lines] (20,-0.5) -- (20,3.5);
\draw[style=help lines] (21,-0.5) -- (21,3.5);
\draw[style=help lines] (24,-0.5) -- (24,3.5);
\draw[style=help lines] (25,-0.5) -- (25,3.5);
\draw[style=help lines] (27,-0.5) -- (27,3.5);

\foreach \x in {0,1,2,3,4,5,6,7,8,9,10,11,12,13,14,15,16,17,18,19,20,21,22,23,24,25,26,27}
	\foreach \y in {0,3,6}
		\fill (\x, \y) circle(1mm);
\end{tikzpicture}
\caption{The z-structure of a TSP tour for the instance $G(28,3)$. The tour shown has the z-vector
$(4,3,3,4,4,3,4,3)$.}
\label{fig:z-structure}
\end{figure*}
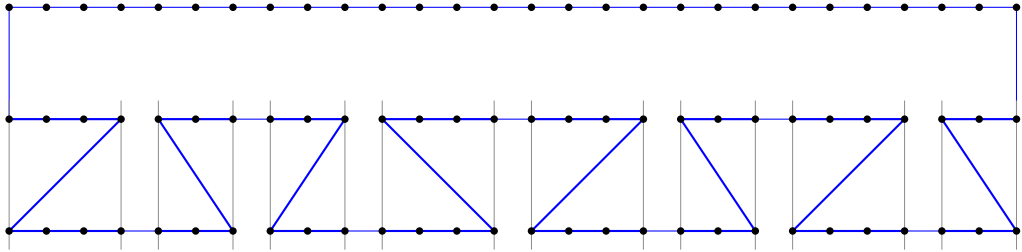

The length of a tour can now easily be computed using its z-vector.
Set 
$$ c(i) ~:=~ 2(i-1) + \sqrt{(i-1)^2 + d^2}.$$ 
Then $c(i)$ is the length of a z-shaped path covering $i$ points in the middle row.

\begin{lemma}
\label{lemma:TotalTourCost}
Let $d\ge 4$ and $k\ge 2$. Then the total length of a TSP tour for $G(n,d)$ 
corresponding to a z-vector $(z_1, z_2, \ldots, z_k)$ is
$$ n + k + 2 d - 2 +  \sum_{i=1}^k c(z_i) .$$
\end{lemma}
\begin{proof}
The length of all the z-shaped subpaths of the tour is $\sum_{i=1}^k c(z_i)$. These $k$ subpaths are connected
by $k-1$ edges of length~1. In addition there is one path of length $n-1$ connecting all the points in the upper row
plus two vertical edges of length $d$ each which connect the upper row with the middle row.
\end{proof}

The next result shows that there always exists an optimum TSP tour for $G(n,d)$ which has a very special structure.

\begin{lemma}\label{lemma:zvector}
Let $d\ge 4$, $k\ge 2$ and $(z_1, z_2, \ldots, z_k)$ be the z-vector of an optimum TSP tour
for $G(n,d)$. Then $z_i \in \{\lfloor n/k \rfloor, \lceil n/k \rceil\}$.
\end{lemma}  
\begin{proof}
Suppose this is not the case. Then there exist $i, j \in \{1,\ldots, k\}$ with $z_i \ge z_j+2$.
As the function $c$ is convex we have $c(z_i -1) + c(z_j+1) < c(z_i) + c(z_j)$.
By Lemma~\ref{lemma:TotalTourCost} this implies that $(z_1, z_2, \ldots, z_k)$ cannot be the
z-vector of an optimum TSP tour for $G(n,d)$.
\end{proof}

The following result shows how to bound the length of a TSP tour for $G(n,d)$ with given z-vector
solely by the length of the z-vector.

\begin{lemma}\label{lemma:TourLength}
Let $d\ge 4$ and $k\ge 2$. Then the length of an optimum TSP tour for $G(n,d)$ corresponding to a z-vector
$(z_1, z_2, \ldots, z_k)$ is at least
$$n+k+2d - 2 + k \cdot c(\frac n k)~ .$$
\end{lemma}  
\begin{proof}
If $(z_1, z_2, \ldots, z_k)$ is an optimum TSP tour for $G(n,d)$ then we may assume that 
$z_1 \le z_2\le \ldots \le z_k$ as reordering the $z_i$-values does not change the length of the tour.
Together with Lemma~\ref{lemma:zvector} this shows that there exists an optimum TSP tour for $G(n,d)$
which corresponds to the z-vector $z_1, \ldots, z_j, z_{j+1}, \ldots, z_k$ with
$z_1 = \ldots = z_j = \lfloor n/k \rfloor$,  $z_{j+1} = \ldots = z_k = \lceil n/k \rceil$, 
and $\sum_{i=1}^k z_i = n$.
As $c$ is convex we have  $\sum_{i=1}^k c(z_i) \ge k\cdot c(\frac n k)$. 
\end{proof}

The next result shows that there always exists an optimum TSP tour for $G(n,d)$ 
which corresponds to an even length z-vector.

\begin{lemma}\label{lemma:evenk}
Let $d\ge 4$ and $n$ be even. Then there exists an optimum TSP tour for $G(n,d)$ which corresponds to an even length
z-vector. Moreover, the length of an optimum TSP tour for $G(n,d)$ is at least:
$$ \min_{1\le k \le n/2} n+2k+2d - 2 + 2k \cdot c(\frac n {2k})~ .$$
\end{lemma}
\begin{proof}
By Lemma~\ref{lemma:TourLength} we only have to rule out that there can exist an optimum TSP tour
for $G(n,d)$ which corresponds to a z-vector of odd length.
By Lemma~\ref{lemma:cheapest_insertion_g1} and its proof we know that if the length of the z-vector
is odd, then it must be~1. 
A TSP tour for $G(n,d)$ corresponding to a z-vector of length one has length
\begin{equation} \label{eq:k=1tourlength}
3d + 3n -4 + \sqrt{(n-2)^2 + d^2}
\end{equation}
A tour corresponding to a z-vector of length~$2$ has length 
\begin{equation} \label{eq:k=2tourlength}
2d + 3n -4 + 2 \sqrt{\left(\frac n2 -1\right)^2 + d^2}
\end{equation}
Now using (\ref{eq:sqrt}) we 
have $$d+\sqrt{(n-2)^2 + d^2} ~\ge~ \sqrt{(n-2)^2 + 4 d^2} ~=~ \sqrt{\left(\frac n2 -1\right)^2 + d^2}~ .$$
This implies (\ref{eq:k=1tourlength})$\ge$(\ref{eq:k=2tourlength}) and therefore the minimum
is always attained for an even length z-vector.
\end{proof}

\begin{theorem}\label{thm:opttourlength}
Let $n\ge 17$ be even. Then an optimum TSP tour for $G(n, \sqrt{n-1})$ has length 
$ 4n  - 4 + 2 \sqrt{n-1}$.
\end{theorem}
\begin{proof}
Let 
\begin{eqnarray*}
f(k,d) &:=& n+2k+2d-2+2k\cdot c\left(\frac{n}{2k}\right)\\
       & = & n+2k+2d-2+2k\cdot \left( 2 \cdot(\frac{n}{2k} - 1) + \sqrt{\left(\frac{n}{2k} - 1\right)^2 +d^2} \right)\\
       & = & 3n-2k+2d-2 + \sqrt{\left(n-2k\right)^2 +4 d^2 k^2} 
\end{eqnarray*}

For $d=\sqrt{n-1}$ we get   
\begin{eqnarray*}
f(k,\sqrt{n-1}) & = & 3n-2k+2\sqrt{n-1}-2 + \sqrt{\left(n-2k\right)^2 +4 (n-1) k^2} \\
                & = & 3n-2k+2\sqrt{n-1}-2 + \sqrt{n^2 +4 n k (k-1)} \\
\end{eqnarray*}

From Lemma~\ref{lemma:evenk} and its proof we know that there exists a TSP tour in 
$G(n, \sqrt{n-1})$ of length $f(1,\sqrt{n-1})$. We now claim that
\begin{equation}\label{eq:tourlength} 
f(k,\sqrt{n-1}) > f(1,\sqrt{n-1}) ~~~ \text{ for } k=2, 3, \ldots, n/2.
\end{equation} 
By Lemma~\ref{lemma:evenk} this finishes the proof.
To prove  (\ref{eq:tourlength}) let $g(k) := f(k,\sqrt{n-1})$.
We will show that $g'(k) > 0$ for $k>1$. 
We have 
\begin{eqnarray*}
g'(k) = \frac{4kn -2n - 2 \sqrt{n^2 +4 n k (k-1)}}{\sqrt{n^2 +4 n k (k-1)}}
\end{eqnarray*}

Now we have
\begin{eqnarray*}
                  (k-1) n         &\ge& k-1    \\
\Rightarrow~~~~~  k^2 n^2 - k n^2 &\ge& nk^2 - n k \\
\Rightarrow~~~~~  4k^2 n^2 + n^2 - 4 k n^2 &\ge& n^2 + 4nk^2 - 4n k \\
\Rightarrow~~~~~ 2kn-n &\ge& \sqrt{n^2 +4 n k (k-1)} \\
\end{eqnarray*}
This proves that $g'(k) \ge 0 $ for $k\ge1$. 
\end{proof}

\begin{lemma}\label{lemma:subtourbound}
The optimum solution of the subtour LP for the instance $G(n,\sqrt{n-1})$ has value
$3n-3 + 3\sqrt{n-1} + \sqrt{n}$.
\end{lemma}
\begin{proof}
An optimum solution to the subtour LP is indicated in Figure~\ref{fig:newsubtourbound}.
It is easily verified that the value is as claimed.
\end{proof}

Now we can state and prove an explicit formula for the integrality ratio of the instances
$G(n,\sqrt{n-1})$.

\begin{theorem}\label{thm:explicitformula}
Let $n\ge 17$ be even.
The integrality ratio of the subtour LP for the instance $G(n,\sqrt{n-1})$
is 
$$ \frac{4n-4+2\sqrt{n-1}} {3n-4 + 3\sqrt{n-1} + \sqrt{n}}~.$$  
\end{theorem}
\begin{proof}
This result immediately follows from Theorem~\ref{thm:opttourlength}
and Lemma~\ref{lemma:subtourbound}.
\end{proof}

Using Theorem~\ref{thm:explicitformula} one obtains for example 
that the integrality ratio of $G(18,\sqrt{17})$ is $1.14$.  
This is much larger than the lower bound
$1.01$ which can be obtained from~(\ref{eq:lowerbound}).

The approach used to prove Theorem~\ref{thm:explicitformula}
can also be applied to other distance functions $d(n)$. 
Choosing $d(n) = \sqrt{n/2-1}$ one obtains an integrality ratio of
$$ \frac{4n-6+2\sqrt{n/2-1}} {3n-4 + 3\sqrt{n/2-1} + \sqrt{n/2}}~.$$
This value is larger than the value stated in Theorem~\ref{thm:explicitformula}.
However, the proof gets a bit more complicated as 
now one has to prove that the function $f(k,\sqrt{n/2-1})$
appearing in the proof of Theorem~\ref{thm:opttourlength}
attains its minimum for $k=2$. 

\section*{Acknowledgement}
We thank Jannik Silvanus for useful discussions and an anonymous referee for pointing out an 
improved formulation of Lemma~\ref{lemma:approxlength}. 
 
\bibliographystyle{plain}

\bibliography{tsp}

\end{document}